\documentclass[journal,onecolumn,11pt]{IEEEtran}

\usepackage{graphicx}
\usepackage{amssymb}
\usepackage{epsfig}
\usepackage{epsf}
\usepackage{hyperref}

\newtheorem{theorem}{Theorem}[section]

\newtheorem{definition}{Definition}[section]
\newtheorem{corollary}{Corollary}[section]

\begin{document}


\title{Behavior of two-level quantum system driven by non-classical inputs}

\author{Abolghasem~Daeichian,
        Farid~Sheikholeslam}


\thanks{Department of Electrical and Computer Engineering , Isfahan University of Technology, Isfahan, Iran.}
\thanks{a.daiechian@ec.iut.ac.ir}
\thanks{This paper is a postprint of a paper submitted to and accepted for publication in
	IET, control theory and applications and is subject to Institution of Engineering and Technology Copyright. The copy of record is
	available at the IET Digital Library}

\maketitle

\begin{abstract}
Two level quantum system (Qubit) and non-classical states of light such as single photon and superposition of coherent state are under special attention in quantum technologies such as quantum computing, quantum communication and quantum computers. So, behavior of two-level system driven by such inputs is important. In this paper, the behavior of two-level quantum system driven by vacuum state, single photon and superposition of coherent state was investigated by assuming Pauli matrices as system operators in quantum filtering equations. The purity of conditioned and unconditioned state are also analyzed when the system is driven by different inputs. The results show that the stochastic master equation dynamic has more information about the status of system than master equation dynamic.
\end{abstract}

\begin{keywords}
Two-level quantum system, Quantum filtering, Single photon, Superposition of coherent states, Stochastic master equation.
\end{keywords}

\section{Introduction}

In the recent years, establishing systems on the theory of quantum, called quantum technology, has drawn much attention of scientists because of more efficiency than classical technology and exclusive properties of quantum systems such as superposition and entanglement ~\cite{LCYK2011,JAG2011,CW2008,SSG2007,MNMB2001,W2009}. A two-level quantum system, Qubit, play a major role in quantum technology ~\cite{WRR2005,YTN2002,KMN2007}. A Qubit, such as spin of electron, has two base state and could be in a superposition of base states. So, it could be utilized as information carrier ~\cite{CZ1995}. The non-classical states of light such as superposition of coherent states which is also known as Schr\"odinger cat state ~\cite{OTLG2006,OJTG2007,OFTG2009} and single photon states ~\cite{EBL2011} has been produced in some experimental architectures and considered in some applications such as quantum computing ~\cite{RGM2003,KLM2001} and perfect secure communication.
State of the system has to be pure in many applications of quantum systems ~\cite{LJL2010,S2009,MI2000,W1994,WM1993}. States that have complete information about status of systems known as pure states. Thus, reaching or stabilizing of an arbitrary pure state ~\cite{HMH1998,WW2001} is considered in some researches.
\\
Due to the fact that states of a quantum system is not measurable, indirect measurement and quantum filter equations is utilized to estimate the states. An appropriate function of estimated states could be fed back through appropriate actuator to govern the system from any initial state to desired state ~\cite{DP2010,WW2001}. In quantum filtering which is an optimal estimator ~\cite{BRJ2007}, an open quantum system interacting with external electromagnetic field such as light. Then, observables of the system could be estimated based on measurements of the scattered or output light. A practical scenario is illustrated in Fig.\ref{Scenario}.
A general approach to filtering problem was developed by Belavkin based on continuous non-demolition quantum measurement ~\cite{BB1991,B1991}. In Belavkin filtering, the input is a quantum white noise with vacuum state or more generally Gaussian state. To date, the basic problem of filtering has been done for Gaussian states such as coherent state fields and squeezed fields ~\cite{GC1985} and recently, some works on the filtering for non-classical states ~\cite{JMNC2011,JMN2011}.

\begin{figure}[htb]
\begin{center}
\includegraphics{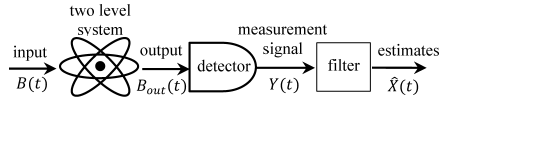}
\caption{Schematic of filtering scenario}
\label{Scenario}
\end{center}
\end{figure}

In this study, the behavior of two-level quantum system driven by different inputs, specially non-classical inputs is investigated. How does a system driven by an input behave, is a basic question in quantum control theory. Therefore a consecution is considering the behavior of two-level quantum system driven by different input such as vacuum state and non-classical inputs such as single photon and superposition of coherent state which was investigated by assuming Pauli matrices as system operators in quantum filtering equations.
A two-level quantum system interacts with an electromagnetic field in the free space and could emit into or absorb from the field depend on the its internal states. The output radiations could be detected by any optical detector setup such as Homodyne detector or photon counting measurement and some information on the internal state of the system can be extracted using measures. Also, this paper is followed by analyzing purity of conditioned and unconditioned state of an open quantum system in general form $G(S,L,H)$. The purity of conditioned and unconditioned state driven by different input is derived and that of state of two-level quantum system is analyzed.

 This paper is structured as follows. Section \ref{Two level quantum system model} devoted to two-level quantum system model to establish the notation and prepare the ground to present behavior of two-level quantum system.
 Behavior of two-level quantum system driven by different inputs, discussions and simulation results have been done in section \ref{Behavior of two-level quantum system}. The simulation results shows that in the situation that the system is derived by multimode single photon which is limited in time, the system will be excited to a mixed state and then it dissipates the absorbed energy and goes back to the ground state. When the input field is in superposition of coherent states, the system converges in steady state to a state which is not ground state. So, the results clarify the possibility of governing the system from any initial state to a desired state by controlling the input field. The purity of conditioned and unconditioned state have been presented in section \ref{Purity of state}. It is proved that purity of conditioned state is more better than unconditioned state. Contrary to the unconditioned state of two-level system, the conditioned state remains pure in the vacuum input and in the steady state of superposition of coherent states input. The paper is concluded in section \ref{Conclusion}.

\section{Two level quantum system model}
\label{Two level quantum system model}
Two level quantum system, like spin of atom, has two base states which are called ground state $|g\rangle$ and excited state $|e\rangle$ that satisfies completeness condition $|g\rangle\langle g|+|e\rangle\langle e|=1$.
Bloch sphere is a geometrical representation of the state space of a two-level quantum mechanical system. The north and south poles are typically chosen to correspond to the standard basis vectors $|g\rangle$ and $|e\rangle$ respectively. The points on the surface of the sphere correspond to the pure states of the system, whereas the interior points correspond to the mixed states.
Any state could be expressed as $|\phi\rangle=c_1|e\rangle+c_2|g\rangle$ where $c_1, c_2\in \mathbf{C}$ and $|c_1|^2+|c_2|^2=1$. Also, any state could be written as density operator
\begin{equation}\label{DensityOperator}
\rho=\frac{1}{2}(I+x\sigma_x+y\sigma_y+z\sigma_z)
\end{equation}
where
 $\sigma_x=\left(\begin{array}{cc}0&1\\1&0\end{array}\right)$,
 $\sigma_y=\left(\begin{array}{cc}0&-i\\i&0\end{array}\right)$ and $\sigma_z=\left(\begin{array}{cc}1&0\\0&-1\end{array}\right)$
which are called Pauli matrices.


An open quantum system is a system interacting with an external environment ~\cite{DS2012}. An open system $G$ coupled to field is determined by thriple $G=(S;L;H)$ where $H$ is the intrinsic Hamiltonian, $L$ is a vector of coupling operators and $S$ is a unitary matrix of operators which is called scattering matrix ~\cite{P1992,HP1984}. The intrinsic Hamiltonian of a two-level system is $H=\frac{\omega}{2}\sigma_z$ where $\omega$ is the atomic frequency of system. Copling operator and scattering matrix are considered $L=\sqrt{\kappa}\sigma_-$ and $S=I$ where $\kappa$ and $\sigma_-=[0\ 0;1\ 0]$ are coupling power and annihilation operator respectively. So, the model of two-level system is
\begin{equation}\label{TwoLevelModel}
G=\left(I, \sqrt{\kappa}\sigma_-, \frac{\omega}{2}\sigma_z\right)
\end{equation}

It is noteworthy that one may take other coupling operators such as $L=\sqrt{\gamma}\sigma_z$ which indicates phase rotation of $\sigma_x$ and $\sigma_y$ operators. Here, we take $L=\sqrt{\gamma}\sigma_-$ to indicate state of $\sigma_z$ operator which has energy level interpretation.

\section{Behavior of two-level quantum system}
\label{Behavior of two-level quantum system}
If the system is considered to be a closed quantum system which means the system has no interaction with ambient like inputs and measurements, there is no dissipation and the system oscillates with its atomic frequency. In this case the dynamic of system is represented by simple Schrodinger equation $|\dot\psi\rangle=-iH|\psi\rangle$ which is analytically solved in physics \cite{G2005}. This case is not application oriented. Henceforth, the system is considered to be an open quantum system.
\\
The evolution (dynamic) of an open quantum system driven by an input field is given by \emph{Quantum Stochastic Differential Equation (QSDE)} that depends on system and field operators. Quantum expectation of QSDE leads to \emph{Master Equation (ME)}. The ME gives the unconditioned dynamics of system. Quantum expectation of QSDE conditioned on measured output is referred as \emph{stochastic master equation (SME)} or \emph{filter}, see ~\cite[Fig.3]{JMNC2011}. The derivation of SME can be done in a rigorous way. The sketch can be found for vacuum input in ~\cite{BRJ2007,B1994} and for single photon and superposition of coherent state inputs in ~\cite{JMNC2011,JMN2011}. In the following, the behavior of two-level quantum system is analyzed by applying filter equations on Pauli matrices as system operators.

\subsection{Vacuum input}
In quantum field theory, the field with no physical particles is called vacuum state which is the quantum state with the lowest possible energy. So, in this part the input is supposed to be the minimum possible input to an open quantum system.
\subsubsection{Homodyne detector case}
The filter equation for the system $G=(S, L, H)$ driven by vacuum state input conditioned on Homodyne measurement is ~\cite{B1994,BRJ2007}:
\begin{equation}\label{FilterVacuumHD}
d\hat{X}(t)=d\pi_t(X)=\pi_t(\mathfrak{L}X)dt+(\pi_t(XL+L^\dag X)-\pi_t(L+L^\dag)\pi_t(X))dW(t).
\end{equation}
where $\pi_t(X)=\hat{X}$ is an estimation of any system operator $X$ conditioned on measurements, $\mathfrak{L}X=-i[X,H]+\frac{1}{2}L^\dag[X,L]+\frac{1}{2}[L^\dag,X]L$, $[A,B]=AB-BA$ and $dW(t)$ is a zero mean Gaussian noise with variance $dt$ which is called \emph{innovation process}.
Always, number averaging over SME gives ME. This fact is used to check the validity of results.

Now, we could proceed to analyse the behavior of two-level system.
\begin{theorem}\label{TheoremQubitVacuumHD}
The dynamics of two-level system (\ref{TwoLevelModel}) driven by field in vacuum state, conditioned on Homodyne detection is:
\begin{eqnarray}\label{QubitVacuumHD}
d\pi_t(\sigma_x)&=&\pi_t(-\omega\sigma_y-\frac{\gamma}{2}\sigma_x)dt+
                    \sqrt{\gamma}\{1+\pi_t(\sigma_z)-\pi_t(\sigma_x)\pi_t(\sigma_x)\}dW(t)\nonumber\\
d\pi_t(\sigma_y)&=&\pi_t(\omega\sigma_x-\frac{\gamma}{2}\sigma_y)dt+
                    \sqrt{\gamma}\{-\pi_t(\sigma_y)\pi_t(\sigma_x)\}dW(t)\\
d\pi_t(\sigma_z)&=&-\gamma\{1+\pi_t(\sigma_z)\}dt+
                    \sqrt{\gamma}\{-\pi_t(\sigma_x)-\pi_t(\sigma_z)\pi_t(\sigma_x)\}dW(t)\nonumber
\end{eqnarray}
where $dW(T)=dY(t)-\sqrt{\gamma}\pi_t(\sigma_x)$.
\end{theorem}
\begin{proof}
Considering (\ref{TwoLevelModel}) and Pauli matrices $\sigma_x$, $\sigma_y$ and $\sigma_z$, Substituting into (\ref{FilterVacuumHD}) results in:
\begin{eqnarray}\label{SigmaxVacuumHD}
d\pi_t(\sigma_x)&=&\pi_t(-i[\sigma_x,\frac{\omega}{2}\sigma_z]+\gamma\sigma_+\sigma_x\sigma_--\frac{1}{2}\gamma(\sigma_+\sigma_-\sigma_x+\sigma_x\sigma_+\sigma_-)dt+\nonumber\\
                 &&\left(\pi_t(\sigma_x\sqrt{\gamma}\sigma_-+\sqrt{\gamma}\sigma_+\sigma_x)-\pi_t(\sqrt{\gamma}(\sigma_-+\sigma_+))\pi_t(\sigma_x)\right)dW(t)
\end{eqnarray}
Pauli matrices, annihilation and creation operators are represented in base states $|g\rangle$ and $|e\rangle$ as $\sigma_x=|g\rangle\langle e|+|e\rangle\langle g|$, $\sigma_y=i(|g\rangle\langle e|-|e\rangle\langle g|)$, $\sigma_z=|e\rangle\langle e|-|g\rangle\langle g|$, $\sigma_-=|g\rangle\langle e|$, $\sigma_+=|e\rangle\langle g|$. Substituting these equations into (\ref{SigmaxVacuumHD}) and doing a bit of calculations using $\langle e|e\rangle=\langle g|g\rangle=1$ and $\langle g|e\rangle=\langle e|g\rangle=0$ will give
\begin{equation}
d\pi_t(\sigma_x)=\pi_t(-\omega\sigma_y-\frac{\gamma}{2}\sigma_x)dt+
                    \sqrt{\gamma}\{1+\pi_t(\sigma_z)-\pi_t(\sigma_x)\pi_t(\sigma_x)\}dW(t)
\end{equation}
The quantum filter for $\sigma_y$ and $\sigma_z$ may now be derived in exactly the same way as was done for the $\sigma_x$.
\end{proof}
The behavior of two-level system driven by input field in vacuum state is illustrated in Fig.\ref{VacuumHD}. The dynamic of any ensemble of system, the dotted (green) line, is given by SME as well as the ensembles average which is given by ME, the solid (red) line. The average of 50 ensemble, the dashed (black) line, is consistent with ME. As the simulation results shows, system from any state goes to ground state because the system dissipates its energy.
\begin{figure}[htb]
\begin{center}
\includegraphics[width=16cm]{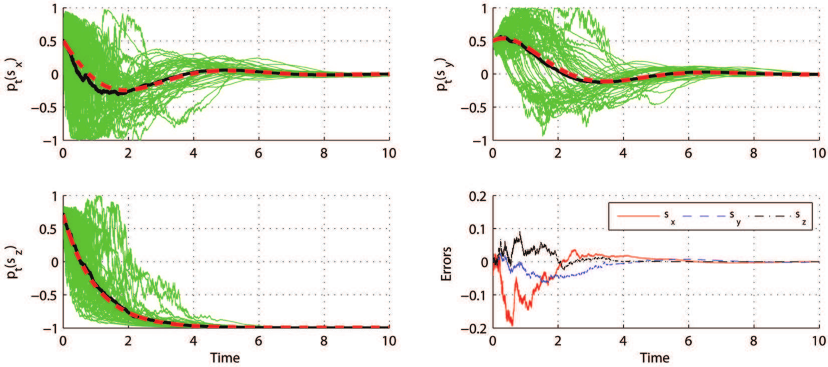}
\caption{Behavior of two-level system driven by input in vacuum state conditioned on Homodyne detection. Gray lines: ensembles of SME; Solid line: average of ensembles; Dashed line: ME}
\label{VacuumHD}
\end{center}
\end{figure}

\subsubsection{Photon detector case}
The filter equation for system $G=(I, L, H)$ driven by vacuum state input, conditioned on photon counting measurement is ~\cite{B1994,BRJ2007}:
\begin{equation}\label{FilterVacuumPD}
d\hat{X}(t)=d\pi_t(X)=\pi_t(\mathfrak{L}X)dt+\left(\frac{\pi_t(L^\dag XL)}{\pi_t(L^\dag L)}-\pi_t(X)\right)dN(t).
\end{equation}
where $dN(t)$ is a Poisson process. We apply this equation to analyse the behavior of two-level system.
\begin{theorem}
The dynamics of two-level system (\ref{TwoLevelModel}) driven by field in vacuum state conditioned on photon detection is:
\begin{eqnarray}\label{QubitVacuumPD}
d\pi_t(\sigma_x)&=&\pi_t(-\omega\sigma_y-\frac{\gamma}{2}\sigma_x)dt
                    -\{\pi_t(\sigma_x)\}dN(t)\nonumber\\
d\pi_t(\sigma_y)&=&\pi_t(\omega\sigma_x-\frac{\gamma}{2}\sigma_y)dt
                    -\{\pi_t(\sigma_y)\}dN(t)\\
d\pi_t(\sigma_z)&=&-\gamma\left(1+\pi_t(\sigma_z)\right)dt
                    -\{1+\pi_t(\sigma_z)\}dN(t)\nonumber
\end{eqnarray}
where $dN(T)=dY(t)-\frac{\sqrt{\gamma}}{2}(1+\pi_t(\sigma_z))dt$.
\end{theorem}
\begin{proof}
Considering (\ref{TwoLevelModel}) and Pauli matrices $\sigma_x$, $\sigma_y$ and $\sigma_z$, Substituting into (\ref{FilterVacuumPD}) results in:
\begin{eqnarray}\label{SigmaxVacuumPD}
d\pi_t(\sigma_x)&=&\pi_t(-i[\sigma_x,\frac{\omega}{2}\sigma_z]+\gamma\sigma_+\sigma_x\sigma_--\frac{1}{2}\gamma(\sigma_+\sigma_-\sigma_x+\sigma_x\sigma_+\sigma_-)dt+\nonumber\\
                 &&\left(\frac{\pi_t(\sigma_+\sigma_x\sigma_-)}{\pi_t(\sigma_+\sigma_-)}-\pi_t(\sigma_x)\right)dN(t)
\end{eqnarray}
continuing in exactly the same way as was done in theorem \ref{TheoremQubitVacuumHD} for Homodyne case.
\end{proof}
The behavior of two-level system driven by input in vacuum state is illustrated in Fig.\ref{VacuumPD}-a. System is supposed to be initially in excited state. So, it is expected that the system collapses to ground state when a photon has been detected. The time that one photon is detected has Poisson distribution with parameter $\frac{\sqrt{\gamma}}{2}(1+\pi_t(\sigma_z))dt$. Simulation result in Fig.\ref{VacuumPD}-a shows that detection of a photon in $t=2.014$ causes transition from $|e\rangle$ to $|g\rangle$.
 As mentioned before, average(the dashed (black) line) of ensembles of system (the dotted (green) line) has to match with ME (the solid (red) line); See Fig.\ref{VacuumPD}-b.

\begin{figure}[htb]
\begin{center}
$\begin{array}{cc}
\includegraphics[width=8cm]{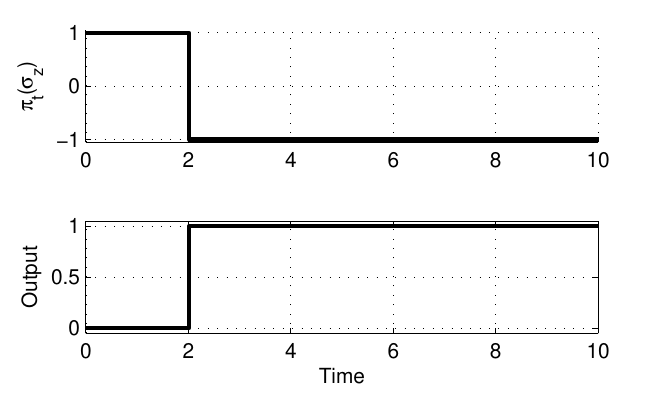}&
\includegraphics[width=8cm]{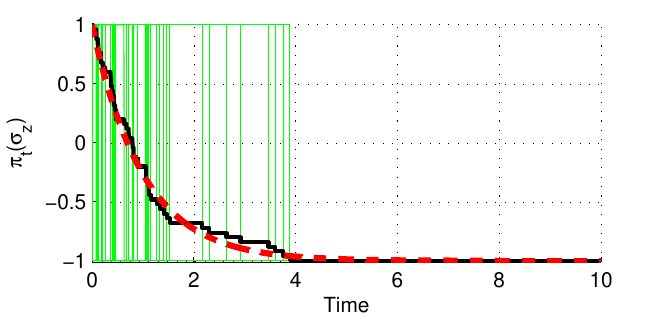}\\
a&b
\end{array}$
\caption{a: Behavior of two-level system driven by input in vacuum state conditioned on photon detection; Column b: Consistency of ME and average of 50 ensembles. Gray lines: ensembles of SME; Solid line: average of ensembles; Dashed line: ME}
\label{VacuumPD}
\end{center}
\end{figure}

\subsection{Single photon input}
Single photon is a non-classical signal. Single photon with one particle state $\xi$ produces by applying creation operator $B^\dag(\xi)=\int_{0}^{\infty}\xi(t)dB^\dag(t)$ on the vacuum state $|0\rangle$; i.e., $|1_\xi\rangle=B^\dag(\xi)|0\rangle$ which is normalized so that $||\xi||^2=\int_0^\infty |\xi(t)|^2dt=1$. In derivation of single photon filter, a generalized density operator $\rho^{jk}=\frac{1}{2}(c^{jk}I+x^{jk}\sigma_x+y^{jk}\sigma_y+z^{jk}\sigma_z)\:(j,k=0,1)$ is defined ~\cite{JMNC2011,JMN2011}. Only $\rho^{11}$ has physical meaning which is the density operator of the system and $\rho^{jk}\:j,k\neq 1$ are defined to calculate $\rho^{11}$ and do not have any physical concept. So, the expectation of any system operator $X$ can be calculated directly from it using $\langle X\rangle=Tr[X\rho_{11}(t)]$.
\subsubsection{Homodyne detector case}
The filter equation for a system $G=(S,L,H)$ driven by single photon field, conditioned on Homodyne measurement is ~\cite{JMNC2011}:
\begin{eqnarray}\label{FilterSinglePhotonHD}
d\pi_t^{11}(X)&=&\{\pi_t^{11}(\mathfrak{L}X)+\pi_t^{01}(S^\dag[X,L])\xi^*(t)+
                \pi_t^{10}([L^\dag,X]S)\xi(t)+\pi_t^{00}(S^\dag XS-X){|\xi(t)|}^2\}dt+\nonumber\\
                &&\{\pi_t^{11}(XL+L^\dag X)+\pi_t^{01}(S^\dag X)\xi^*(t)+\pi_t^{10}(XS)\xi(t)-\pi_t^{11}(X)K_t\}dW(t)\nonumber\\
d\pi_t^{10}(X)&=&\{\pi_t^{10}(\mathfrak{L}X)+\pi_t^{00}(S^\dag[X,L])\xi^*(t)\}dt+\nonumber\\
                &&\{\pi_t^{10}(XL+L^\dag X)+\pi_t^{00}(S^\dag X)\xi^*(t)-\pi_t^{10}(X)K_t\}dW(t)\nonumber\\
d\pi_t^{01}(X)&=&\{\pi_t^{01}(\mathfrak{L}X)+\pi_t^{00}([L^\dag,X]S)\xi(t)\}dt+\nonumber\\
                &&\{\pi_t^{01}(XL+L^\dag X)+\pi_t^{00}(XS)\xi(t)-\pi_t^{01}(X)K_t\}dW(t)\nonumber\\
d\pi_t^{00}(X)&=&\{\pi_t^{00}(\mathfrak{L}X)\}dt+\{\pi_t^{00}(XL+L^\dag X)-\pi_t^{00}(X)K_t\}dW(t)
\end{eqnarray}
where $K_t=\pi_t^{11}(L+L^\dag)+\pi_t^{01}(S)\xi(t)+\pi_t^{10}(S^\dag)\xi^*(t)$ and $dW(t)=dY(t)-K_tdt$. Eq.(\ref{FilterSinglePhotonHD}) shows that $\pi_t^{01}(X)=\pi_t^{10}(X)$. Thus, in the following only $\pi_t^{01}(X)$ is considered. Due to the definition of generalized density operator it is required to calculate $c^{jk}=\pi_t^{jk}(I)$ so we consider $\sigma_x$, $\sigma_y$, $\sigma_z$ and $I$ as system operators.
\begin{theorem}
The dynamics of two-level system (\ref{TwoLevelModel}) driven by field in single photon state, conditioned on Homodyne detection is:
\begin{eqnarray}\label{QubitSinglePhotonHD}
d\pi_t^{11}(\sigma_x)&=&\{\pi_t^{11}(-\omega\sigma_y-\frac{\gamma}{2}\sigma_x)+
                       \pi_t^{01}(\sqrt{\gamma}\sigma_z)\xi^*(t)+
                       \pi_t^{10}(\sqrt{\gamma}\sigma_z)\xi(t)\}dt+\nonumber\\
                       &&\{\sqrt{\gamma}\pi_t^{11}(I+\sigma_z)+\pi_t^{01}(\sigma_x)\xi^*(t)+
                       \pi_t^{10}(\sigma_x)\xi(t)-\pi_t^{11}(\sigma_x)K_t\}dW(t)\nonumber\\
d\pi_t^{10}(\sigma_x)&=&\{\pi_t^{10}(-\omega\sigma_y-\frac{\gamma}{2}\sigma_x)+
                       \pi_t^{00}(\sqrt{\gamma}\sigma_z)\xi^*(t)\}dt+\nonumber\\
                       &&\{\sqrt{\gamma}\pi_t^{10}(I+\sigma_z)+\pi_t^{00}(\sigma_x)\xi^*(t)
                       -\pi_t^{10}(\sigma_x)K_t\}dW(t)\nonumber\\
d\pi_t^{00}(\sigma_x)&=&\{\pi_t^{00}(-\omega\sigma_y-\frac{\gamma}{2}\sigma_x)\}dt+
                       \{\sqrt{\gamma}\pi_t^{00}(I+\sigma_z)-\pi_t^{00}(\sigma_x)K_t\}dW(t)\nonumber\\
d\pi_t^{11}(\sigma_y)&=&\{\pi_t^{11}(\omega\sigma_x-\frac{\gamma}{2}\sigma_y)+
                       \pi_t^{01}(-i\sqrt{\gamma}\sigma_z)\xi^*(t)+
                       \pi_t^{10}(i\sqrt{\gamma}\sigma_z)\xi(t)\}dt+\nonumber\\
                       &&\{\pi_t^{01}(\sigma_y)\xi^*(t)+
                       \pi_t^{10}(\sigma_y)\xi(t)-\pi_t^{11}(\sigma_y)K_t\}dW(t)\nonumber\\
d\pi_t^{10}(\sigma_y)&=&\{\pi_t^{10}(\omega\sigma_x-\frac{\gamma}{2}\sigma_y)+
                       \pi_t^{00}(-i\sqrt{\gamma}\sigma_z)\xi^*(t)\}dt+\nonumber\\
                       &&\{\pi_t^{00}(\sigma_y)\xi^*(t)-\pi_t^{10}(\sigma_y)K_t\}dW(t)\nonumber\\
d\pi_t^{00}(\sigma_y)&=&\{\pi_t^{00}(\omega\sigma_x-\frac{\gamma}{2}\sigma_y)\}dt+
                       \{-\pi_t^{00}(\sigma_y)K_t\}dW(t)\nonumber\\
d\pi_t^{11}(\sigma_z)&=&\{-\gamma\pi_t^{11}(I+\sigma_z)+\pi_t^{01}(-\sqrt{\gamma}(\sigma_x-i\sigma_y)\xi^*(t)+
                       \pi_t^{10}(-\sqrt{\gamma}(\sigma_x+i\sigma_y)\xi(t)\}dt+\nonumber\\
                       &&\{\sqrt{\gamma}\pi_t^{11}(-\sigma_x)+\pi_t^{01}(\sigma_z)\xi^*(t)+
                       \pi_t^{10}(\sigma_z)\xi(t)-\pi_t^{11}(\sigma_z)K_t\}dW(t)\nonumber\\
d\pi_t^{10}(\sigma_z)&=&\{-\gamma\pi_t^{10}(I+\sigma_z)+\pi_t^{00}(-\sqrt{\gamma}(\sigma_x-i\sigma_y)\xi^*(t)\}dt+\nonumber\\
                       &&\{\sqrt{\gamma}\pi_t^{10}(-\sigma_x)+\pi_t^{00}(\sigma_z)\xi^*(t)
                       -\pi_t^{10}(\sigma_z)K_t\}dW(t)\nonumber\\
d\pi_t^{00}(\sigma_z)&=&\{-\gamma\pi_t^{00}(I+\sigma_z)\}dt+
                       \{\sqrt{\gamma}\pi_t^{00}(-\sigma_x)-\pi_t^{00}(\sigma_z)K_t\}dW(t)\nonumber\\
d\pi_t^{10}(I)&=&\{\sqrt{\gamma}\pi_t^{10}(\sigma_x)+\pi_t^{00}(I)\xi^*(t)-\pi_t^{10}(I)K_t\}dW(t)\nonumber\\
d\pi_t^{00}(I)&=&\{\sqrt{\gamma}\pi_t^{00}(\sigma_x)-\pi_t^{00}(I)K_t\}dW(t)
\end{eqnarray}
where $K_t=\sqrt{\gamma}\pi_t^{11}(\sigma_x)+\pi_t^{01}(I)\xi(t)+\pi_t^{10}\xi^*(t)$ and $dW(t)=dY(t)-K_tdt$.
\end{theorem}
\begin{proof}
This theorem would be proved in the same way as was done for theorem \ref{TheoremQubitVacuumHD}.
\end{proof}
We take the single photon shape to be Gaussian same as ~\cite{JMNC2011}
\begin{equation}\label{PhotonShape}
\xi(t)=\left(\frac{\Omega^2}{2\pi}\right)^{0.25}\exp{[\frac{-\Omega^2}{4}(t-t_c)^2]}
\end{equation}
where $\Omega$ and $t_c$ specify the frequency bandwidth of the pulse and the peak arrival time respectively. The simulation result of dynamical equations (\ref{QubitSinglePhotonHD}) driven by single photon (\ref{PhotonShape}) is illustrated in Fig.\ref{SinglePhotonHD}. As before, the consistency of the trajectories is confirmed with the master equation solution by calculating a numerical average of the trajectories. It is worth noting that $Pe_{10}$ and $Pe_{01}$ represent the coherency between $|e\rangle$ and $|g\rangle$. These are complex numbers in each time that the amplitude is plotted for each ensemble. Due to those are with positive and negative numbers, the average of ensembles is close to zero.

\begin{figure}[htb]
\begin{center}
\includegraphics[width=16cm]{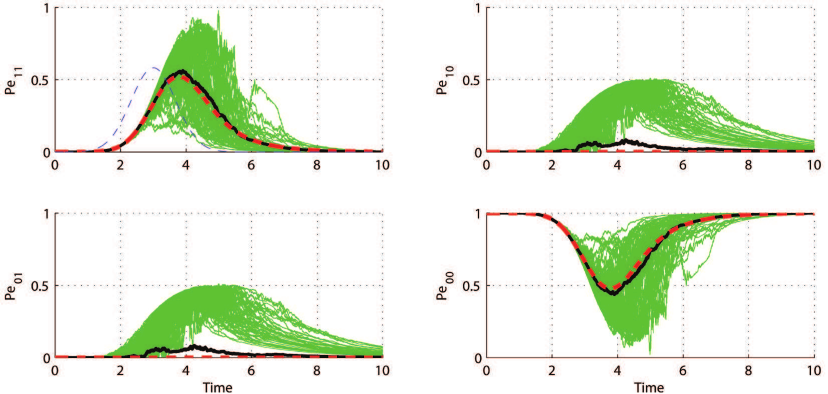}
\caption{Behavior of two-level system driven by input in single photon conditioned on Homodyne detection. Gray lines: ensembles of SME; Solid line: average of ensembles; Dashed line: ME; Thin-dashed line in $Pe_{11}$: Photon shape}
\label{SinglePhotonHD}
\end{center}
\end{figure}

\subsubsection{Photon detector case}
The filter equation for system $G=(S, L, H)$ driven by single photon field, conditioned on photon counting is ~\cite{JMNC2011}:
\begin{eqnarray}\label{FilterSinglePhotonPD}
d\pi_t^{11}(X)&=&\{\pi_t^{11}(\mathfrak{L}X)+\pi_t^{01}(S^\dag[X,L])\xi^*(t)+
                \pi_t^{10}([L^\dag,X]S)\xi(t)+\pi_t^{00}(S^\dag XS-X){|\xi(t)|}^2\}dt+\nonumber\\
                &&\{\nu_t^{-1}\left(\pi_t^{11}(L^\dag XL)+\pi_t^{01}(S^\dag XL)\xi^*(t)+\pi_t^{10}(L^\dag XS)\xi(t)+\pi_t^{00}(S^\dag XS){|\xi(t)|}^2\right)-\nonumber\\
                &&\pi_t^{11}(X)\}dN(t)\nonumber\\
d\pi_t^{10}(X)&=&\{\pi_t^{10}(\mathfrak{L}X)+\pi_t^{00}(S^\dag[X,L])\xi^*(t)\}dt+\nonumber\\
                &&\{\nu_t^{-1}\left(\pi_t^{10}(L^\dag XL)+\pi_t^{00}(S^\dag XL)\xi^*(t)\right)-\pi_t^{10}(X)\}dN(t)\nonumber\\
d\pi_t^{01}(X)&=&\{\pi_t^{01}(\mathfrak{L}X)+\pi_t^{00}([L^\dag,X]S)\xi(t)\}dt+\nonumber\\
                &&\{\nu_t^{-1}\left(\pi_t^{01}(L^\dag XL)+\pi_t^{00}(L^\dag XS)\xi(t)\right)-\pi_t^{01}(X)\}dN(t)\nonumber\\
d\pi_t^{00}(X)&=&\{\pi_t^{00}(\mathfrak{L}X)\}dt+\{\nu_t^{-1}\left(\pi_t^{00}(L^\dag XL)\right)-\pi_t^{00}(X)\}dN(t)
\end{eqnarray}
where $\nu_t=\pi_t^{11}(L^\dag L)+\pi_t^{10}(L^\dag S)\xi(t)+\pi_t^{01}(S^\dag L)\xi^*(t)+\pi_t^{00}(I){|\xi(t)|}^2$ and $dN(t)=dY(t)-\nu_tdt$.
\begin{theorem}
The dynamics of two-level system (\ref{TwoLevelModel}) driven by field in single photon state, conditioned on photon detection is:
\begin{eqnarray}\label{QubitSinglePhotonPD}
d\pi_t^{11}(\sigma_x)&=&\{\pi_t^{11}(-\omega\sigma_y-\frac{\gamma}{2}\sigma_x)+
                       \pi_t^{01}(\sqrt{\gamma}\sigma_z)\xi^*(t)+
                       \pi_t^{10}(\sqrt{\gamma}\sigma_z)\xi(t)\}dt+ \nonumber\\
                       &&\{\nu_t^{-1}\left(\sqrt{\gamma}\pi_t^{01}(I+\sigma_z)\xi^*(t)+
                       \sqrt{\gamma}\pi_t^{10}(I+\sigma_z)\xi(t)+
                       \pi_t^{00}(\sigma_x){|\xi(t)|}^2\right)-\pi_t^{11}(\sigma_x)\}dN(t)\nonumber\\
d\pi_t^{10}(\sigma_x)&=&\{\pi_t^{10}(-\omega\sigma_y-\frac{\gamma}{2}\sigma_x)+
                       \pi_t^{00}(\sqrt{\gamma}\sigma_z)\xi^*(t)\}dt+\nonumber\\
                       &&\{\nu_t^{-1}\left(\sqrt{\gamma}\pi_t^{00}(I+\sigma_z)\xi^*(t)\right)
                       -\pi_t^{10}(\sigma_x)\}dN(t)\nonumber\\
d\pi_t^{00}(\sigma_x)&=&\{\pi_t^{00}(-\omega\sigma_y-\frac{\gamma}{2}\sigma_x)\}dt+
                       \{-\pi_t^{00}(\sigma_x)\}dN(t)\nonumber\\
d\pi_t^{11}(\sigma_y)&=&\{\pi_t^{11}(\omega\sigma_x-\frac{\gamma}{2}\sigma_y)+
                       \pi_t^{01}(-i\sqrt{\gamma}\sigma_z)\xi^*(t)+
                       \pi_t^{10}(i\sqrt{\gamma}\sigma_z)\xi(t)\}dt+\nonumber\\
                       &&\{\nu_t^{-1}\left(-i\sqrt{\gamma}\pi_t^{01}(I+\sigma_z)\xi^*(t)+
                       i\sqrt{\gamma}\pi_t^{10}(I+\sigma_z)\xi(t)+
                       \pi_t^{00}(\sigma_y){|\xi(t)|}^2\right)-\pi_t^{11}(\sigma_y)\}dN(t)\nonumber\\
d\pi_t^{10}(\sigma_y)&=&\{\pi_t^{10}(\omega\sigma_x-\frac{\gamma}{2}\sigma_y)+
                       \pi_t^{00}(-i\sqrt{\gamma}\sigma_z)\xi^*(t)\}dt+\nonumber\\
                       &&\{\nu_t^{-1}\left(-i\sqrt{\gamma}\pi_t^{00}(I+\sigma_z)\xi^*(t)\right)
                       -\pi_t^{10}(\sigma_y)\}dN(t)\nonumber\\
d\pi_t^{00}(\sigma_y)&=&\{\pi_t^{00}(\omega\sigma_x-\frac{\gamma}{2}\sigma_y)\}dt+
                       \{-\pi_t^{00}(\sigma_y)\}dN(t)\nonumber\\
d\pi_t^{11}(\sigma_z)&=&\{-\gamma\pi_t^{11}(I+\sigma_z)+\pi_t^{01}(-\sqrt{\gamma}(\sigma_x-i\sigma_y))\xi^*(t)+
                       \pi_t^{10}(-\sqrt{\gamma}(\sigma_x+i\sigma_y))\xi(t)\}dt+\nonumber\\
                       &&\{\nu_t^{-1}\left(-\frac{\gamma}{2}\pi_t^{11}(I+\sigma_z)-
                       \sqrt{\gamma}\pi_t^{01}(\sigma_x-i\sigma_y)\xi^*(t)-
                       \sqrt{\gamma}\pi_t^{10}(\sigma_x+i\sigma_y)\xi(t)+
                       \pi_t^{00}(\sigma_z){|\xi(t)|}^2\right)\nonumber\\
                       &&-\pi_t^{11}(\sigma_z)\}dN(t)\nonumber\\
d\pi_t^{10}(\sigma_z)&=&\{-\gamma\pi_t^{10}(I+\sigma_z)+\pi_t^{00}(-\sqrt{\gamma}(\sigma_x-i\sigma_y)\xi^*(t)\}dt+\nonumber\\
                       &&\{\nu_t^{-1}\left(-\frac{\gamma}{2}\pi_t^{10}(I+\sigma_z)-
                       \sqrt{\gamma}\pi_t^{01}(\sigma_x-i\sigma_y)\xi^*(t)\right)
                       -\pi_t^{10}(\sigma_z)\}dN(t)\nonumber\\
d\pi_t^{00}(\sigma_z)&=&\{-\gamma\pi_t^{00}(I+\sigma_z)\}dt+
                       \{\nu_t^{-1}\left(-\frac{\gamma}{2}\pi_t^{00}(I+\sigma_z)\right)-
                       \pi_t^{00}(\sigma_z)\}dN(t)\nonumber\\
d\pi_t^{10}(I)&=&\{\nu_t^{-1}\left(\frac{\gamma}{2}\pi_t^{10}(I+\sigma_z)+
                \frac{\sqrt{\gamma}}{2}\pi_t^{00}(\sigma_x-i\sigma_y)\xi^*(t)\right)
                -\pi_t^{10}(I)\}dN(t)\nonumber\\
d\pi_t^{00}(I)&=&\{\nu_t^{-1}\left(\frac{\gamma}{2}\pi_t^{00}(I+\sigma_z)\right)-\pi_t^{00}(I)\}dN(t)
\end{eqnarray}
where $\nu_t=\frac{\gamma}{2}(I+\pi_t^{11}(\sigma_z))+
\frac{\sqrt{\gamma}}{2}(\pi_t^{01}(\sigma_x-i\sigma_y)\xi^*(t)+
\frac{\sqrt{\gamma}}{2}(\pi_t^{10}(\sigma_x+i\sigma_y)\xi^(t)+\pi_t^{00}(I){|\xi(t)|}^2$ and $dN(t)=dY(t)-\nu_tdt$.
\end{theorem}
\begin{proof}
In the same way as theorem \ref{TheoremQubitVacuumHD}.
\end{proof}
The dynamics of (\ref{QubitSinglePhotonPD}) driven by single photon shape (\ref{PhotonShape}) is illustrated in Fig.\ref{SinglePhotonPD}. The system is supposed to be initially in ground state.
When the system is derived by single photon which is limited in time, the system has transient response that is excited to a mixed state and then dissipates the absorbed energy and goes back to the ground state as shown in Figs.\ref{SinglePhotonHD},\ref{SinglePhotonPD}. So, There are some atoms that may become fully excited by a single photon input.

\begin{figure}[htb]
\begin{center}
\includegraphics[width=16cm]{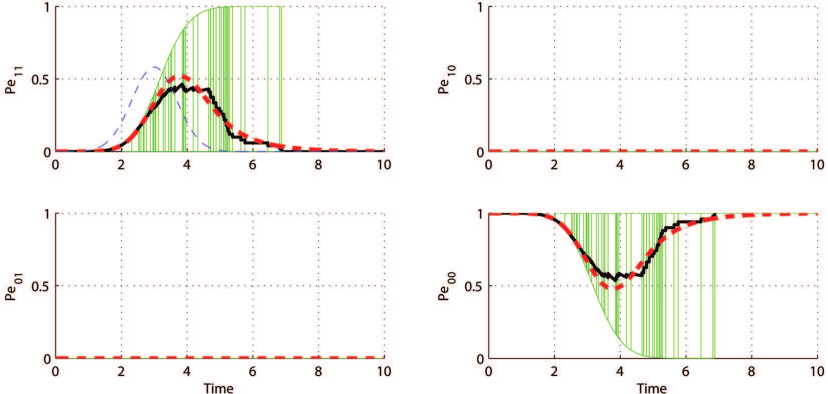}
\caption{Behavior of two-level system driven by input in single photon conditioned on photon detection. Gray lines: ensembles of SME; Solid line: average of ensembles; Dashed line: ME; Thin-dashed line in $Pe_{11}$: Photon shape}
\label{SinglePhotonPD}
\end{center}
\end{figure}

\subsection{Superposition of coherent state input}
Coherent state in continuous mode is given by applying Displacement or Weyl operator to the vacuum state
\begin{equation}
|\alpha\rangle=D|0\rangle
\end{equation}
where $\alpha$ is a given function. So, superposition of coherent state could be written as $|\psi\rangle=\sum_{j=1}^ns_j|\alpha_j\rangle$ which must satisfy normalization condition $\sum_{i,j}s_i^*s_jg_{ij}=1$ where $g_{ij}=\langle\alpha_i|\alpha_j\rangle=\exp(-\frac{1}{2}{||\alpha_i||}^2-\frac{1}{2}{||\alpha_j||}^2+\langle\alpha_i,\alpha_j\rangle)$ and $\langle\alpha_i,\alpha_j\rangle=\int_{-\infty}^{\infty}\alpha_i^*(s)\alpha_j(s)ds$. The estimation could be written as $\pi_t(X)=\sum_{i,j}\pi_t^{ij}(X)$.

\subsubsection{Homodyne detector case}
The filter equation for system $G=(S,L,H)$ driven by superposition of coherent field, conditioned on Homodyne detection is ~\cite{JMNC2011}:
\begin{eqnarray}\label{FilterSuperpositionHD}
d\pi_t^{ij}(X)      &=&\pi_t^{ij}(\mathcal{G}^{ij}X)dt+\mathcal{H}_t^{ij}(X)dW(t)
\end{eqnarray}
This is an It\^o stochastic equation. The operator for drift term is $\mathcal{G}^{ij}X   =\mathfrak{L}X+S^\dag[X,L]\alpha_i^*(t)+[L^\dag,X]S\alpha_j(t)+(S^\dag XS-X)\alpha_i^*(t)\alpha_j(t)$ while that for diffusion term is $\mathcal{H}_t^{ij}X =\pi_t^{ij}(XL+L^\dag X+XS\alpha_j(t)+S^\dag\alpha_i^*(t)X)-\pi_t^{ij}\sum_l\frac{s_l^*s_l}{N_a}\pi_t^{ll}(L+L^\dag+S\alpha_l(t)+S^\dag\alpha_l^*(t))$ where $N_a=\sum_is_i^*s_i$. Also in this case, we require to consider $\pi_t^{ij}(I)$ as a system operator.
\begin{theorem}
The dynamics of two-level system (\ref{TwoLevelModel}) driven by field in superposition of coherent state, conditioned on Homodyne detection is:
\begin{eqnarray}\label{QubitSuperpositionHD}
d\pi_t^{ij}(\sigma_x)&=&\{-\omega\pi_t^{ij}(\sigma_y)-\frac{\gamma}{2}\pi_t^{ij}(\sigma_x)+
                        \sqrt{\gamma}\pi_t^{ij}(\sigma_z)(\alpha_i^*(t)+\alpha_j(t))\}dt+\nonumber\\
                        &&\{\sqrt{\gamma}\pi_t^{ij}(I+\sigma_z)+\pi_t^{ij}(\sigma_x)(\alpha_j(t)+\alpha_i^*(t))-\nonumber\\
                        &&\pi_t^{ij}(\sigma_x)\sum_l\frac{s_l^*s_l}{N_a}
                        \left(\sqrt{\gamma}\pi_t^{ll}(\sigma_x)+\pi_t^{ll}(I)(\alpha_l(t)+\alpha_l^*(t))\right)\}dW(t)\nonumber\\
d\pi_t^{ij}(\sigma_y)&=&\{\omega\pi_t^{ij}(\sigma_x)-\frac{\gamma}{2}\pi_t^{ij}(\sigma_y)-i
                        \sqrt{\gamma}\pi_t^{ij}(\sigma_z)(\alpha_i^*(t)-\alpha_j(t))\}dt+\nonumber\\
                        &&\{\pi_t^{ij}(\sigma_y)(\alpha_j(t)+\alpha_i^*(t))-
                        \pi_t^{ij}(\sigma_y)\sum_l\frac{s_l^*s_l}{N_a}
                        \left(\sqrt{\gamma}\pi_t^{ll}(\sigma_x)+\pi_t^{ll}(I)(\alpha_l(t)+\alpha_l^*(t))\right)\}dW(t)\nonumber\\
d\pi_t^{ij}(\sigma_z)&=&\{-\gamma\pi_t^{ij}(I+\sigma_z)-
                        \sqrt{\gamma}\pi_t^{ij}(\sigma_x-i\sigma_y)\alpha_i^*(t)-\sqrt{\gamma}\pi_t^{ij}(\sigma_x+i\sigma_y)\alpha_j(t)\}dt+\nonumber\\
                        &&\{\sqrt{\gamma}\pi_t^{ij}(-\sigma_x)+\pi_t^{ij}(\sigma_z)(\alpha_j(t)+\alpha_i^*(t))-\nonumber\\
                        &&\pi_t^{ij}(\sigma_z)\sum_l\frac{s_l^*s_l}{N_a}
                        \left(\sqrt{\gamma}\pi_t^{ll}(\sigma_x)+\pi_t^{ll}(I)(\alpha_l(t)+\alpha_l^*(t))\right)\}dW(t)\nonumber\\
d\pi_t^{ij}(I)&=&\{\sqrt{\gamma}\pi_t^{ij}(\sigma_x)+\pi_t^{ij}(I)(\alpha_j(t)+\alpha_i^*(t))-\nonumber\\
                        &&\pi_t^{ij}(I)\sum_l\frac{s_l^*s_l}{N_a}
                        \left(\sqrt{\gamma}\pi_t^{ll}(\sigma_x)+\pi_t^{ll}(I)(\alpha_l(t)+\alpha_l^*(t))\right)\}dW(t)
\end{eqnarray}
\end{theorem}
\begin{proof}
May be proved in the same way as was done for theorem \ref{TheoremQubitVacuumHD}.
\end{proof}
Simulation result is illustrated in Fig.\ref{SuperpositionHD} when the input is taken the schr\"odinger cat state $|\psi\rangle=\frac{1}{\sqrt{2}}|\alpha\rangle+\frac{1}{\sqrt{2}}|-\alpha\rangle$ where $\alpha$ supposed to be a pulse between $t=0$ and $t=5$ with amplitude equal to 1.

\begin{figure}[htb]
\begin{center}
\includegraphics[width=16cm]{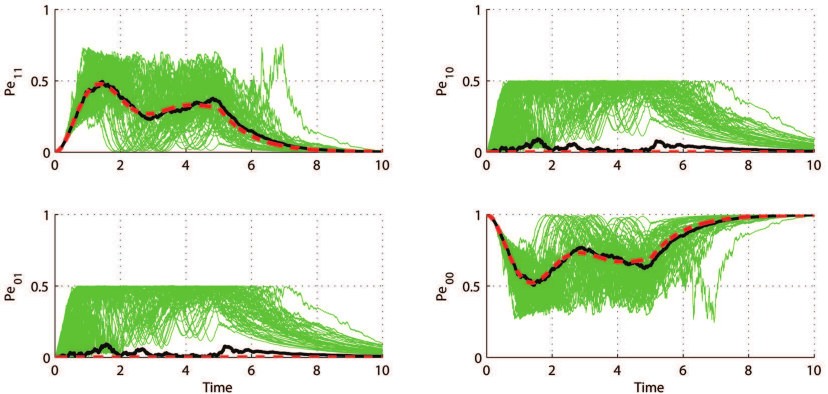}
\caption{Behavior of two-level system driven by input in Superposition of coherent state conditioned on Homodyne detection. Gray lines: ensembles of SME; Solid line: average of ensembles; Dashed line: ME.}
\label{SuperpositionHD}
\end{center}
\end{figure}

\subsubsection{Photon detector case}
The filter equation for system $G=(S,L,H)$ driven by superposition of coherent field, conditioned on photon counting is ~\cite{JMNC2011}:
\begin{eqnarray}\label{FilterSuperpositionPD}
d\pi_t^{ij}(X)      &=& \pi_t^{ij}(\mathcal{G}^{ij}X)dt+\mathcal{H}_t^{ij}(X)dN(t)\nonumber\\
\mathcal{H}_t^{ij}X &=& \frac{\pi_t^{ij}(L^\dag XL+\alpha_j(t)L\dag XS+
                        \alpha_i^*(t)S^\dag XL)+\alpha_i^*(t)\alpha_j(t)S^\dag XS}
                        {\sum_l\frac{s_l^*s_l}{N_a}\pi_t^{ll}(L^\dag L+\alpha_l(t)L^\dag S+
                        \alpha_l^*(t)S^\dag L+||\alpha_l||^2I)}-\pi_t^{ij}(X)\nonumber\\
\mathcal{G}^{ij}X   &=& \mathfrak{L}X+S^\dag[X,L]\alpha_i^*(t)+[L^\dag,X]S\alpha_j(t)+
                        (S^\dag XS-X)\alpha_i^*(t)\alpha_j(t)
\end{eqnarray}
\begin{theorem}
The dynamics of two-level system (\ref{TwoLevelModel}) driven by field in superposition of coherent state conditioned on photon detection is:
\begin{eqnarray}\label{QubitSuperpositionPD}
d\pi_t^{ij}(\sigma_x)&=&\{-\omega\pi_t^{ij}(\sigma_y)-\frac{\gamma}{2}\pi_t^{ij}(\sigma_x)+
                        \sqrt{\gamma}\pi_t^{ij}(\sigma_z)(\alpha_i^*(t)+\alpha_j(t))\}dt+\nonumber\\
                        &&\{\frac{\frac{\sqrt{\gamma}}{2}\pi_t^{ij}(I+\sigma_z)(\alpha_j(t)+
                        \alpha_i^*(t))+\pi_t^{ij}(\sigma_x)(\alpha_j(t)\alpha_i^*(t))}
                        {\mathcal{M}X}-\pi_t^{ij}(\sigma_x)\}dN(t)\nonumber\\
d\pi_t^{ij}(\sigma_y)&=&\{\omega\pi_t^{ij}(\sigma_x)-\frac{\gamma}{2}\pi_t^{ij}(\sigma_y)-i
                        \sqrt{\gamma}\pi_t^{ij}(\sigma_z)(\alpha_i^*(t)-\alpha_j(t))\}dt+\nonumber\\
                        &&\{\frac{\frac{\sqrt{\gamma}}{2}\pi_t^{ij}(I+\sigma_z)i(\alpha_j(t)-
                        \alpha_i^*(t))+\pi_t^{ij}(\sigma_y)(\alpha_j(t)\alpha_i^*(t))}
                        {\mathcal{M}X}-\pi_t^{ij}(\sigma_y)\}dN(t)\nonumber\\
d\pi_t^{ij}(\sigma_z)&=&\{-\gamma\pi_t^{ij}(I+\sigma_z)-
                        \sqrt{\gamma}\pi_t^{ij}(\sigma_x-i\sigma_y)\alpha_i^*(t)-\sqrt{\gamma}\pi_t^{ij}(\sigma_x+i\sigma_y)\alpha_j(t)\}dt+\nonumber\\
                        &&\{\frac{\frac{\gamma}{2}\pi_t^{ij}(I+\sigma_z)-
                        \frac{\sqrt{\gamma}}{2}\pi_t^{ij}(\sigma_x+i\sigma_y)\alpha_j(t)-
                        \frac{\sqrt{\gamma}}{2}\pi_t^{ij}(\sigma_x-i\sigma_y)\alpha_i^*(t)+
                        \pi_t^{ij}(\sigma_z)\alpha_j(t)\alpha_i^*(t)}{\mathcal{M}X}-\nonumber\\
                        &&\pi_t^{ij}(\sigma_z)\}dN(t)\nonumber\\
d\pi_t^{ij}(I)       &=&\{\frac{\frac{\gamma}{2}\pi_t^{ij}(I+\sigma_z)-
                        \frac{\sqrt{\gamma}}{2}\pi_t^{ij}(\sigma_x+i\sigma_y)\alpha_j(t)-
                        \frac{\sqrt{\gamma}}{2}\pi_t^{ij}(\sigma_x-i\sigma_y)\alpha_i^*(t)+
                        \pi_t^{ij}(I)\alpha_j(t)\alpha_i^*(t)}{\mathcal{M}X}-\nonumber\\
                        &&\pi_t^{ij}(I)\}dN(t)
\end{eqnarray}
where $\mathcal{M}X=\sum_l\frac{s_l^*s_l}{N_a}\left(\frac{\gamma}{2}\pi_t^{ll}(I+\sigma_z)+
                        \frac{\sqrt{\gamma}}{2}\pi_t^{ll}(\sigma_x+i\sigma_y)\alpha_l(t)+
                        \frac{\sqrt{\gamma}}{2}\pi_t^{ll}(\sigma_x-i\sigma_y)\alpha_l^*(t)+
                        ||\alpha_l(t)||^2\pi_t^{ll}(I)\right)$.
\end{theorem}
\begin{proof}
May be proved in the same way as was done for theorem \ref{TheoremQubitVacuumHD}.
\end{proof}
Simulation result is illustrated in Fig.\ref{SuperpositionPD} when the input is taken the schr\"odinger cat state $|\psi\rangle=\frac{1}{\sqrt{2}}|0\rangle+\frac{1}{\sqrt{2}}|-\alpha\rangle$ where $\alpha$ supposed to be a pulse between $t=0$ and $t=5$ with amplitude equal to 1.

\begin{figure}[htb]
\begin{center}
\includegraphics[width=16cm]{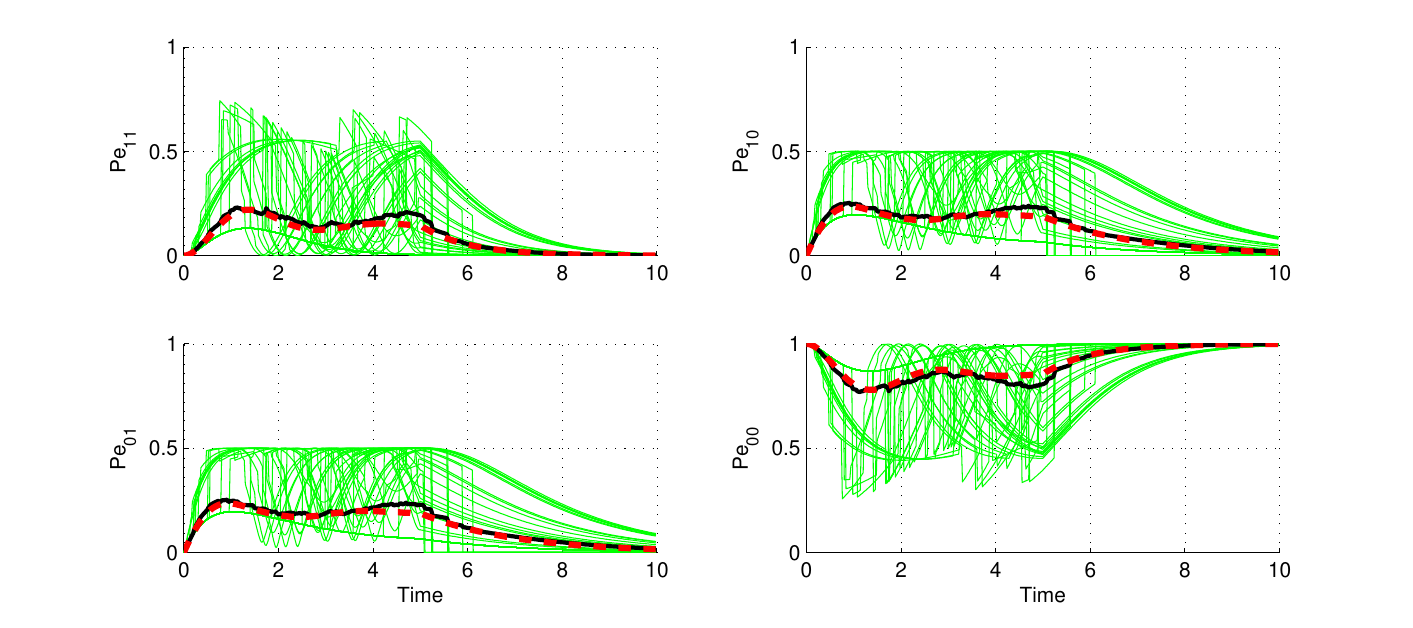}
\caption{Behavior of two-level system driven by input in Superposition of coherent state conditioned on photon detection. Gray lines: ensembles of SME; Solid line: average of ensembles; Dashed line: ME.}
\label{SuperpositionPD}
\end{center}
\end{figure}

When the input field is in superposition of coherent states, the system converges to a state which is not ground state as shown by the simulation results in Figs.\ref{SuperpositionHD},\ref{SuperpositionPD}. The system dissipates the absorbed energy and goes back to the ground state After finishing the input pulse.

\section{Purity of state}
\label{Purity of state}
Pure states have many applications in quantum technologies ~\cite{LJL2010,S2009,MI2000,W1994,WM1993}. Thus, reaching or stabilizing of an arbitrary pure state is considered by some researchers ~\cite{HMH1998,WW2001}. In this section the purity of conditioned and unconditioned states of system $G(S,L,H)$ are derived and then we are focused on purity of state of two-level system driven by different inputs.
$Tr[\rho^2]$ is a measure of purity. So, The purity is defined as ~\cite{WWM2001}:
\begin{definition}
Let $\rho$ be the density operator of the system. Then, the purity of state is defined as $P=2Tr[\rho^2]-1$.
\end{definition}
The system is pure if $P=1$ and fully mixed if $P=0$. The dynamical equation that governs purity could be written as $dP=2dTr[\rho^2]$.

\subsection{General case}
\label{General case}
In the following the dynamical equation of purity for unconditioned and average of conditioned states are derived.
\begin{theorem}
The dynamics of purity of unconditioned state of the system $G(S,L,H)$ is
\begin{itemize}
\item If input is in vacuum state:
    \begin{eqnarray}\label{TrRho2Vacuum}
        \frac{d}{dt}P&=&4Tr[[\rho,L]\rho L^\dag]
    \end{eqnarray}
\item If input is  in single photon state:
    \begin{eqnarray}\label{TrRho2SinglePhoton}
        \frac{d}{dt}P&=&4Tr[[\rho_{11},L]\rho_{11} L^\dag+
                               [L^\dag,\rho_{11}]S\rho_{01}\xi(t)+
                               [\rho_{11},L]\rho_{10}S^\dag\xi^*(t)+
                               (S\rho_{00}S^\dag-\rho_{00})\rho_{11}|\xi^2(t)|
                              ]
    \end{eqnarray}
\item If input is in superposition of coherent state:
    \begin{eqnarray}\label{TrRho2Superposition}
        \frac{d}{dt}P&=&2\sum_{mnjk}s_m^*s_ns_j^*s_kTr[\rho_{mn}L^\dag[\rho_{jk},L]+\rho_{jk}L^\dag[\rho_{mn},L]+\nonumber\\
                    &&[L^\dag,\rho_{mn}]S\rho_{jk}\alpha_j+[\rho_{mn},L]\rho_{jk}S^\dag\alpha_k^*+
                      \rho_{mn}(S\rho_{jk}S^\dag-\rho_{jk})\alpha_j\alpha_k^*+\nonumber\\
                    &&[L^\dag,\rho_{jk}]S\rho_{mn}\alpha_m+[\rho_{jk},L]\rho_{mn}S^\dag\alpha_n^*+
                      \rho_{jk}(S\rho_{mn}S^\dag-\rho_{mn})\alpha_m\alpha_n^*]
    \end{eqnarray}
\end{itemize}
\end{theorem}
\begin{proof}
The case which input is in vacuum state is proved here. Other cases may be proved in the same way. The ME of density operator could be derived by considering ME of operator $X$, i.e. $E[d\pi_t(X)]=d\bar{\omega}_t(X)=\bar{\omega}_t(\mathfrak{L}X)dt$, and using $Tr[X(t)\rho]=Tr[X\rho(t)]$. The result is given in of ~\cite[Eq.6]{JMNC2011}:
\begin{equation}
    \frac{d}{dt}\rho(t)=-i[H,\rho(t)]+L\rho L^\dag-\frac{1}{2}(L^\dag L\rho+\rho L^\dag L)
\end{equation}
 Now, using It\^o product rule
\begin{equation}\label{ItoProductRule}
    dTr[\rho^2]=Tr[\rho d\rho+d\rho\rho+d\rho d\rho]
\end{equation}
suppose that $dt^2=0$ results in
\begin{equation}
    \frac{d}{dt}Tr[\rho^2]=2Tr[-i\rho[H,\rho]+\rho L\rho L^\dag-\frac{1}{2}(\rho L^\dag L\rho+\rho\rho L^\dag L)]
\end{equation}
using the cyclic property of trace will give Eq.(\ref{TrRho2Vacuum}). This complete the proof.
\end{proof}

\begin{theorem}
The dynamics of average purity of state of the system $G(S,L,H)$ conditioned on Homodyne detection is
\begin{itemize}
\item If input is in vacuum state:
    \begin{eqnarray}\label{CondTrRho2Vacuum}
        \frac{d}{dt}P&=&4Tr[[\rho,L]\rho L^\dag]+Tr[(L+L^\dag)\rho L\rho-
                        4Tr[(L+L^\dag)\rho]\rho L\rho+(L+L^\dag)\rho L^\dag\rho-\nonumber\\
                        &&2Tr[(L+L^\dag)\rho]\rho L^\dag\rho]+(Tr[L+L^\dag]\rho)^2Tr[\rho^2].
    \end{eqnarray}
\item If input is in single photon state:
    \begin{eqnarray}\label{CondTrRho2SinglePhoton}
        \frac{d}{dt}P&=&4Tr[\rho_{11}(L+L^\dag)\rho_{11}(L+L^\dag)+
                        (L+L^\dag)(\rho_{11}S\rho_{01}\xi(t)+\rho_{10}S^\dag\rho_{11}\xi^*(t))+\nonumber\\
                        &&(\rho_{10}S^\dag+S\rho_{01})(L\rho_{11}+\rho_{11}L^\dag)\xi(t)+
                        (\rho_{10}S^\dag\xi^*(t))^2+(S\rho_{01}\xi(t))^2+\nonumber\\
                        &&(2\rho_{10}\rho_{01}+S\rho_{00}S^\dag\rho_{11}-
                        \rho_{00}\rho_{11})\rho_{11}|\xi^2(t)|-\nonumber\\
                        &&2K_t\rho_{11}((L+L^\dag)\rho_{11}+\rho_{10}S^\dag\xi*(t)+S\rho_{01}\xi(t))+
                        (K_t\rho_{11})^2
                        ]
    \end{eqnarray}
where $K_t=Tr[(L+L^\dag)\rho_{11}+S\rho_{01}\xi(t)+S^\dag\rho_{10}\xi^*(t)]$
\item If input is in superposition of coherent state:
    \begin{eqnarray}\label{CondTrRho2Superposition}
        \frac{d}{dt}P&=&2\sum_{mnjk}s_m^*s_ns_j^*s_kTr[\rho_{jk}(L+L^\dag)\rho_{mn}(L+L^\dag)+\nonumber\\
                        &&\rho_{jk}\rho_{mn}L^\dag(S\alpha_k-S\alpha_j-K_t)+
                        L\rho_{mn}\rho_{jk}(S^\dag\alpha_j^*-S^\dag\alpha_k^*-K_t)+\nonumber\\
                        &&\rho_{mn}\rho_{jk}L^\dag(S\alpha_n-S\alpha_m-K_t)+
                        L\rho_{jk}\rho_{mn}(S^\dag\alpha_m^*-S^\dag\alpha_n^*-K_t)+\nonumber\\
                        &&\rho_{mn}S^\dag\rho_{jk}(L\alpha_n^*+L^\dag\alpha_m^*+S\alpha_m\alpha_n^*+S^\dag\alpha_m^*\alpha_j^*-K_t\alpha_m^*)+\nonumber\\
                        &&\rho_{mn}S\rho_{jk}(L\alpha_k+L^\dag\alpha_j+S\alpha_n\alpha_k+S^\dag\alpha_j\alpha_k^*-K_t\alpha_k)+\nonumber\\
                        &&\rho_{jk}S^\dag\rho_{mn}(L\alpha_k^*+L^\dag\alpha_j^*-K_t\alpha_j^*)+
                          \rho_{jk}S\rho_{mn}(L\alpha_n+L^\dag\alpha_m-K_t\alpha_n)]
    \end{eqnarray}
where $K_t=\sum_l\frac{s_l^*s_l}{N_a}Tr[(L+L^\dag+S\alpha_l(t)+S^\dag\alpha_l^*(t))\rho_{ll}]$.
\end{itemize}

\end{theorem}
\begin{proof}
Here, the input in vacuum state is proved. Other cases may be proved in the same way. The SME of density operator could be derived by considering SME of operator $X$, Eq.(\ref{FilterVacuumHD}), and using $Tr[X(t)\rho]=Tr[X\rho(t)]$ which is given in ~\cite[Eq.14]{JMNC2011}:
\begin{equation}
    d\rho(t)=\{-i[H,\rho(t)]+L\rho L^\dag-\frac{1}{2}(L^\dag L\rho+\rho L^\dag L)\}dt+
                        \{L\rho+L^\dag\rho-Tr[(L+L^\dag)\rho]\rho\}dW(t)
\end{equation}
Now, using It\^o product rule (Eq.\ref{ItoProductRule})
substituting $dt^2=0$, $dtdW=0$ and $dW^2=dt$ results in
\begin{equation}
    E[\frac{d}{dt}Tr[\rho^2]]=2Tr[-i\rho[H,\rho]+\rho L\rho L^\dag-
                    \frac{1}{2}(\rho L^\dag L\rho+\rho\rho L^\dag L)+
                    (L\rho+L^\dag\rho-Tr[(L+L^\dag)\rho]\rho)^2]
\end{equation}
using the cyclic property of trace will give Eq.(\ref{CondTrRho2Vacuum}). This complete the proof.
\end{proof}

A similar theorem could be written for the average purity of state conditioned on photon detector. To be brief, that one is omitted.

\subsection{Two-level system case}
In section \ref{General case} the purity of conditioned and unconditioned state of quantum system $G(S,L,H)$ driven by input in vacuum state, single photon and superposition of coherent state was presented that may eventuate to notable results in special cases. In the following, the purity of state of a two-level system is analysed. Purity of state of two-level system could be written in Bloch representation vector $(x,y,z)$ as $P=x^2+y^2+z^2$. $(E[\pi_t(\sigma_x)],E[\pi_t(\sigma_y)],E[\pi_t(\sigma_z)])$ and $(\pi_t(\sigma_x),\pi_t(\sigma_y),\pi_t(\sigma_z))$ are considered as Bloch vector for unconditioned and conditioned dynamics respectively.

\begin{theorem}\label{QubitPurityME}
The dynamical equations for unconditioned state of the system (\ref{TwoLevelModel}) (i.e., the expectation of Eqs.(\ref{QubitVacuumHD}), (\ref{QubitSinglePhotonHD}) and (\ref{QubitSuperpositionHD}) for input in vacuum state, single photon state and superposition of coherent state respectively) are
\begin{itemize}
\item If input is in vacuum state:
    \begin{eqnarray}\label{PurityQubitVacuumME}
        \frac{d}{dt}P&=&-2\gamma(P+z^2+2z)
    \end{eqnarray}
\item If input is in single photon state:
    \begin{eqnarray}\label{PurityQubitSinglePhotonME}
        \frac{d}{dt}P&=&-2\gamma(P+z_{11}^2+2z_{11})+4\sqrt{\gamma}Re\{\left((x_{11}+iy_{11})z_{10}-
                        (x_{10}+iy_{10})z_{11}\right)\xi(t)\}
    \end{eqnarray}
\item If input is in superposition of coherent state:
    \begin{eqnarray}\label{PurityQubitSuperpositionME}
        \frac{d}{dt}P&=&2\sum_{mnjk}s_m^*s_ns_j^*s_k\{-\gamma(P+z_{ij}z_{mn}+2z_{ij})+\nonumber\\
                       &&2\sqrt{\gamma}z_{mn}\left((x_{ij}+iy_{ij})(\alpha_n-\alpha_j)+
                       (x_{ij}-iy_{ij})(\alpha_m^*-\alpha_i^*)\right)\}
    \end{eqnarray}
\end{itemize}
\end{theorem}
\begin{proof}
Consider input is in vacuum state. Substituting Bloch representation of density operator and coupling matrix into (\ref{TrRho2Vacuum}):

\begin{eqnarray}
\frac{d}{dt}P &=& \gamma Tr[[1+x\sigma_x+y\sigma_y+z\sigma_z,\sigma_-](1+x\sigma_x+y\sigma_y+z\sigma_z)\sigma_+]
\end{eqnarray}
using $[\sigma_x,\sigma_y]=2i\sigma_z$, $[\sigma_y,\sigma_z]=2i\sigma_x$, $[\sigma_z,\sigma_x]=2i\sigma_y$, $\sigma_-=\frac{\sigma_x-i\sigma_y}{2}$ and $\sigma_+=\frac{\sigma_x+i\sigma_y}{2}$ together with definition of $P$ establishes Eq.(\ref{PurityQubitVacuumME}). Other cases may be proved in the same way by using Eqs.(\ref{TrRho2SinglePhoton}) and (\ref{TrRho2Superposition}).
\end{proof}
The mixed state could be interpreted as loss of knowledge on the state of system. In unconditioned case, knowledge about the system will be lost when it is driven by an input Because of no measurement is applied on the system. Theorem \ref{QubitPurityME} shows that the complete purity is obtained after the input goes away and the system backs to $z=-1$ ($P=1$ and $z=1$ is an equilibrium of Eqs.(\ref{PurityQubitVacuumME}, \ref{PurityQubitSinglePhotonME}, \ref{PurityQubitSuperpositionME}) when the inputs go away).

\begin{theorem}
The dynamical equations for the state of the system (\ref{TwoLevelModel}) conditioned on Homodyne detection (i.e., the Eqs.(\ref{QubitVacuumHD}), (\ref{QubitSinglePhotonHD}) and (\ref{QubitSuperpositionHD}) for input in vacuum state, single photon state and superposition of coherent state respectively) are
\begin{itemize}
\item If input is in vacuum state:
    \begin{eqnarray}\label{PurityQubitVacuumHD}
        \frac{d}{dt}P&=&2\gamma(P-1)(x^2-1)
    \end{eqnarray}
\item If input is in single photon state:
    \begin{eqnarray}\label{PurityQubitSinglePhotonHD}
        \frac{d}{dt}P&=&2(K_t^2-\gamma)P+2\gamma(1+x_{11}^2-2\sqrt{\gamma}x_{11}K_t)+\nonumber\\
                       &&4Re\{(x_{10}^2+y_{10}^2+z_{10}^2)\xi^2\}+4(|x_{01}|+|y_{01}|+|z_{01}|)|\xi|^2+\nonumber\\
                       &&4Re\{\sqrt{\gamma}x_{10}-x_{10}x_{11}K_t-y_{10}y_{11}K_t-z_{10}z_{11}K_t+
                       i\sqrt{\gamma}y_{11}z_{10}-i\sqrt{\gamma}z_{11}y_{10}\}
    \end{eqnarray}
\item If input is in superposition of coherent state:
    \begin{eqnarray}\label{PurityQubitSuperpositionHD}
        \frac{d}{dt}P&=&2\sum_{mnjk}s_m^*s_ns_j^*s_k\{(P+c_{ij}c_{mn})
                       \left((\alpha_n+\alpha_m^*)(\alpha_j+\alpha_i^*)
                       -K_t(\alpha_n+\alpha_m^*+\alpha_j+\alpha_i^*)+K_t^2\right)-\nonumber\\
                       &&\gamma(P-c_{ij}c_{mn})
                       +2\gamma x_{ij}x_{mn}-2\sqrt{\gamma}K_t(c_{ij}x_{mn}+x_{ij}c_{mn})\}
    \end{eqnarray}
\end{itemize}
\end{theorem}
\begin{proof}
This theorem could be proved in the same way as theorem(\ref{QubitPurityME}) by using Eqs.(\ref{CondTrRho2Vacuum}-\ref{CondTrRho2Superposition}). Thus, the proof is omitted.
\end{proof}

When the system is observed, we get knowledge about the system and it make sense that purity of state increase. Such a behavior can not be seen through unconditional dynamics. There are some notable facts in conditioned case:

First,
\begin{corollary}
Eq.(\ref{PurityQubitVacuumHD}) shows that if the system is initially in a pure state then the conditioned dynamics remain pure despite of applying vacuum state input.
\end{corollary}

Second, simulation results of dynamical equations (\ref{PurityQubitSinglePhotonHD}, \ref{PurityQubitSuperpositionHD}) show that ME dynamic gives more loss of knowledge than SME, see Fig.\ref{TrVacuumHD}.

Third, simulation results of Eq.(\ref{PurityQubitSuperpositionHD}) shows that ME dynamic gives an steady state loss of knowledge but SME dynamic gives steady state purity when $\alpha(t)$ is taken to be a constant, see Fig.\ref{TrVacuumHD}-c. It means that better control could be applied by using SME dynamics to reach or stabilize a pure state.

A controller was designed in \cite{WW2001,HMH1998} for the situation that the input field is a coherent field. Hofmann and Mahler \cite{HMH1998} have been used the master equation to design the controller. They showed that any state in bottom-half of Bloch sphere could be stabilized. But, Wang and Wiseman \cite{WW2001} used the stochastic master equation and proved that any points on the Bloch sphere, but not the equator, of the sphere could be stabilized. The better stabilization results of ~\cite{WW2001} than ~\cite{HMH1998} is due to the fact that SME gives more information about state of the system.

\begin{figure}[htb]
\begin{center}
$\begin{array}{ccc}
\includegraphics[width=5.6cm]{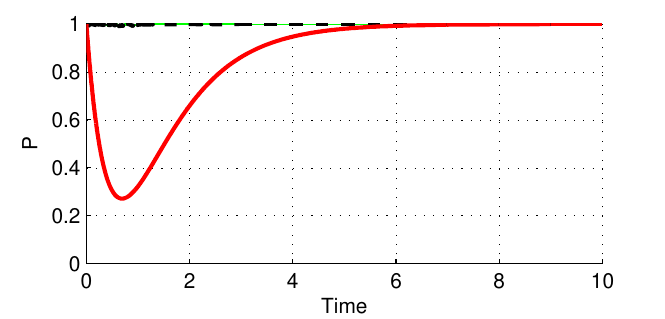}&
\includegraphics[width=5.6cm]{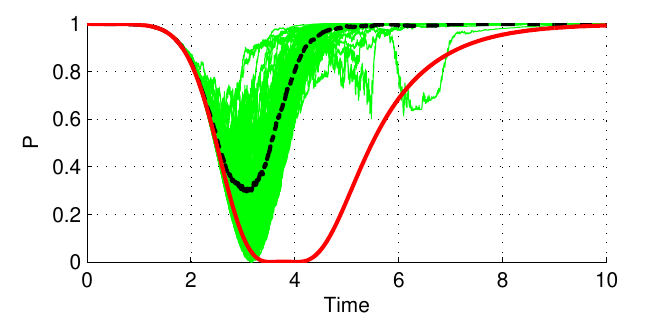}&
\includegraphics[width=5.6cm]{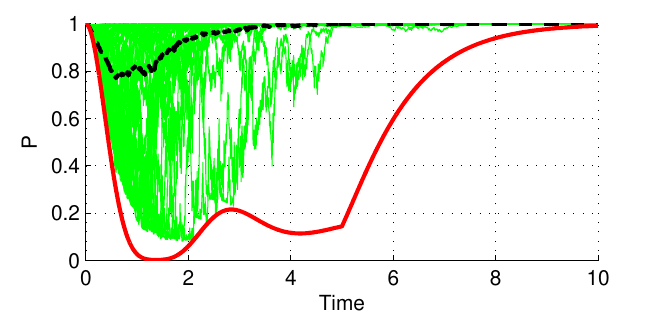}\\
\end{array}$\\
$\begin{array}{ccc}
\includegraphics[width=5.6cm]{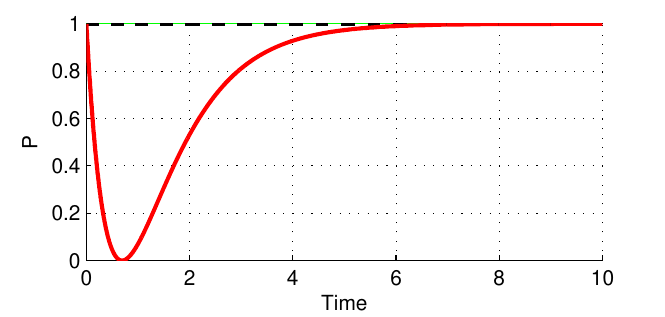}&
\includegraphics[width=5.6cm]{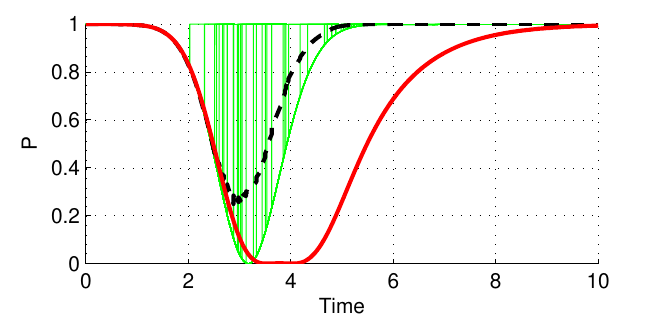}&
\includegraphics[width=5.6cm]{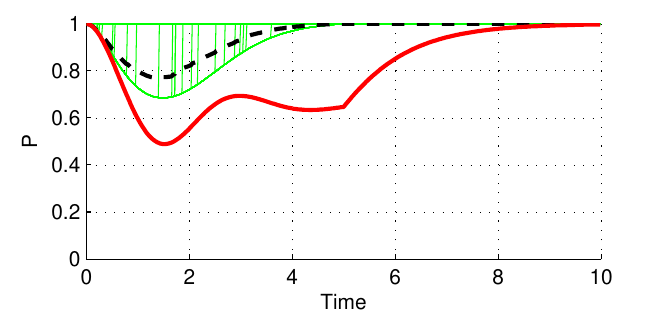}\\
a&b&c
\end{array}$
\caption{Purity of conditioned and unconditioned state. First row: conditioned on Homodyne detection; Second row: conditioned on Photon detection. Column a: vacuum input; Column b: single photon input; Column c: superposition of coherent states input. Gray lines: ensembles of SME; Solid line: average of ensembles; Dashed line: ME}
\label{TrVacuumHD}
\end{center}
\end{figure}

\section{Conclusion}
\label{Conclusion}
In this paper, the behavior of two-level quantum system driven by different input, specially non-classical inputs, was investigated. Pauli matrices are considered as system operators and filter equations are applied to estimate state of system conditioned on Homodyne and photon detection. In addition, the purity of state was analyzed in general form. Furthermore, the case of a two-level system is considered. it shows that if the system is initially in a pure state then SME dynamic remains pure. Also, the simulation results shows that SME dynamic gives states that are more pure than ME dynamic. The ME dynamic predicts loss of knowledge. So, using SME dynamic gives more information about the system than ME.

\section{Acknowledgment}
\label{Acknowledgment}
The first author would like to thank M.R. James for his great helps, comments and suggestions during the visiting period at ANU and also A.R. Carvalho and M.R. Hush for their guidance on simulating quantum filtering.


\end{document}